\documentclass[twoside,english]{iopart}
\usepackage[T1]{fontenc}
\usepackage[latin9]{inputenc}
\usepackage{geometry}
\expandafter\let\csname equation*\endcsname\relax
\expandafter\let\csname endequation*\endcsname\relax 
\usepackage{amsthm}
\usepackage{amsmath}
\usepackage{amssymb}
\usepackage{esint}

\makeatletter
 
 \@ifundefined{textcolor}{}
 {%
   \definecolor{BLACK}{gray}{0}
   \definecolor{WHITE}{gray}{1}
   \definecolor{RED}{rgb}{1,0,0}
   \definecolor{GREEN}{rgb}{0,1,0}
   \definecolor{BLUE}{rgb}{0,0,1}
   \definecolor{CYAN}{cmyk}{1,0,0,0}
   \definecolor{MAGENTA}{cmyk}{0,1,0,0}
   \definecolor{YELLOW}{cmyk}{0,0,1,0}
 }
  \theoremstyle{plain}
  \newtheorem{lem}{\protect\lemmaname}
\theoremstyle{plain}
\newtheorem{thm}{\protect\theoremname}
  \theoremstyle{plain}
  \newtheorem{cor}{\protect\corollaryname}
\newtheorem{prop}{\protect\propositionname}
\theoremstyle{plain}
\makeatother

\usepackage{babel}
  \providecommand{\lemmaname}{Lemma}
\providecommand{\corollaryname}{Corollary}
\providecommand{\theoremname}{Theorem}
\providecommand{\propositionname}{Proposition}

\begin{document}



\title[Multidimensional $q$-Cramér-Rao inequalities]{On multidimensional generalized Cramér-Rao inequalities, uncertainty
relations and characterizations of generalized $q$-Gaussian distributions
\footnote{This is a preprint version that differs from the published version, J. Phys. A: Math. Theor. 46, 095303, 2013 doi:10.1088/1751-8113/46/9/095303, in minor revisions, pagination and typographics details.} 
}

\author{J.-F. Bercher}

\ead{jean-francois.bercher@univ-paris-est.fr}


\selectlanguage{english}%

\address{Laboratoire d'informatique Gaspard Monge, UMR 8049  ESIEE-Paris,
Université Paris-Est}

\date{\today}
\begin{abstract}

%

In the present work, we show how the generalized Cram\'er-Rao inequality
for the estimation of a parameter, presented in a recent paper, can
be extended to the mutidimensional case with general norms on $\mathbb{R}^{n}$,
and to a wider context. As a particular case, we obtain a new multidimensional Cramér-Rao
inequality which is saturated by generalized $q$-Gaussian distributions.
We also give another related Cram\'er-Rao inequality, for a general norm,
which is saturated  as well by these distributions. Finally, we derive
uncertainty relations from these Cramér-Rao inequalities. These uncertainty
relations involve moments computed with respect to escort distributions,
and we show that some of these relations are saturated by generalized
$q$-Gaussian distributions. These results introduce extended versions
of Fisher information, new Cram\'er-Rao inequalities, and new characterizations of generalized $q$-Gaussian
distributions which are important in several areas of physics and
mathematics.

\end{abstract}




\maketitle

\section{Introduction}

It is well known that the Gaussian distribution is a central distribution
with respect to classical information measures and inequalities. In
particular, the Gaussian distribution is both a maximum entropy and
a minimum Fisher information distribution over all distributions with
the same variance. We will show that the same kind of result holds
for the family of generalized $q$-Gaussians, for Rényi or Tsallis
entropy and a suitable extension of the Fisher information. 

These generalized $q$-Gaussians appear in statistical physics, where
they are the maximum entropy distributions of the nonextensive thermostatistics
\cite{tsallis_introduction_2009}. The generalized $q$-Gaussian distributions
define a versatile family that can describe problems with compact
support as well as problems with heavy tailed distributions. They
are also analytical solutions of actual physical problems, see \cite{lutz_anomalous_2003}, \cite{schwaemmle_q-gaussians_2008}
\cite{vignat_isdetection_2009}, \cite{ohara_information_2010}, and
are sometimes known as Barenblatt-Pattle functions, following their
identification by \cite{barenblatt_unsteady_1952}, \cite{pattle_diffusion_1959}.
We shall also mention that the Generalized $q$-Gaussian distributions
appear in other fields, namely as the solution of non-linear diffusion
equations, or as the distributions that saturate some sharp inequalities
in functional analysis \cite{del_pino_best_2002}, \cite{del_pino_optimal_2003},
\cite{cordero-erausquin_mass-transportation_2004}, \cite{agueh_sharp_2008}.

In the literature, and in particular within nonextensive thermostatistics,
several extensions of the Fisher information and of the Cramér-Rao
inequality have been proposed, e.g. \cite{pennini_semiclassical_2007},
\cite{furuichi_maximum_2009}, \cite{furuichi_generalized_2010}, \cite{naudts_generalised_2008}, \cite{naudts_q-exponential_2009}.
In information theory, the remarkable work by Lutwak et al. \cite{lutwak_Cramer_2005},
\cite{lutwak_extensions_2012} also defines an extended Fisher information
and a Cramér-Rao inequality saturated by $q$-Gaussian distributions.
 However the Fisher information is originally defined in a broader
context as the information about a parameter of a parametric family
of distributions. It is only in the special case of a location parameter
that it reduces to the Fisher information of the distribution. The
Fisher information is especially important for the formulation of
the Cramér-Rao inequality. This well-known inequality appears in the
context of estimation theory, where it defines a lower bound on the
variance of any estimator of a parameter.

In our recent work \cite{bercher_generalized_2012}, we have thrown
a bridge between concepts in estimation theory and tools of nonextensive
thermostatistics. Using the notion of escort distribution, we have
established an extended version of the Cramér-Rao inequality for the
estimation of a parameter. This new Cramér-Rao inequality includes
the standard one, as well as Barankin-Vajda versions \cite[Corollary 5.1]{barankin_locally_1949},\cite{vajda_-divergence_1973}
as particular cases. Furthermore, in the case of a location parameter,
we have obtained extended versions of the standard Cramér-Rao inequality,
which are saturated by the generalized $q$-Gaussians. This means
that among all distributions with a given moment, the generalized
$q$-Gaussians are also the minimizers of extended versions of the
Fisher information, just as the standard Gaussian minimizes Fisher
information over all distributions with a given variance. This result
yields a new information-theoretic characterization of these generalized
Gaussian distributions. 

However, a quite frustrating point is that these results were limited
to the univariate case, while the multidimensional case is obviously
of high importance. This restriction is overcome in the present paper
where we show that previous results can be extended to the multidimensional
case. More than that, we consider an even wider context where moments
of the error are computed with respect to two different probability
distributions. In addition, giving our results for general norms will
be hardly more difficult than for Euclidean norms, so we consider
this general case from the beginning. Finally, we derive new uncertainty relations 
from the multidimensional Cram\'er-Rao inequalities. Let us now give a brief overview
of the results, together with the organization of the paper.

Let $\theta\in\Theta\subseteq\mathbb{R}^{n}$ be a multidimensional
parameter that we wish to estimate using data $x.$ We show that for
$\hat{\theta}(x)$ an estimator of $\theta$, if $f(x;\theta)$ and
$g(x;\theta)$ are two probability densities, and if $\alpha$ and
$\beta$ are Hölder conjugates of each other, then 
\begin{equation}
E\left[\left\Vert \hat{\theta}(x)-\theta\right\Vert ^{\alpha}\right]^{\frac{1}{\alpha}}I_{\beta}[f|g;\theta]^{\frac{1}{\beta}}\geq\left|n+\nabla_{\theta}.\, B_{f}(\theta)\right|,\label{eq:VeryGeneralCR}
\end{equation}
 where $\|.\|$ is a general norm on $\mathbb{R}^{n},$ $\nabla_{\theta}.\, B_{f}(\theta)$
represents the divergence of the bias between $\hat{\theta}(x)$ and
$\theta$, and $I_{\beta}[f|g;\theta]$ stands for a generalized Fisher
information that measures the information in $f$\, about $\theta$,
and is taken with respect to $g$. This general result is established
in section \ref{sec:The-generalized-Cram=0000E9r-Rao}. Then we discuss
in subsection \ref{sub:Main-consequences-of} some special cases of
this general inequality. In particular, if $f$ and $g$ is a pair
of $q$-escort distributions, we obtain 
\begin{alignat}{1}
E\left[\left\Vert \hat{\theta}(x)-\theta\right\Vert ^{\alpha}\right]^{\frac{1}{\alpha}}\, I_{\beta,q}\left[f|g;\theta\right]^{\frac{1}{\beta}} & \geq\left|n+\nabla_{\theta}.E_{q}\left[\hat{\theta}(x)-\theta\right]\right|\label{eq:GenqCR1-1}\\
E_{\bar{q}}\left[\left\Vert \hat{\theta}(x)-\theta\right\Vert ^{\alpha}\right]^{\frac{1}{\alpha}}\, I_{\beta,q}\left[f|g;\theta\right]^{\mbox{\ensuremath{\frac{1}{\beta}}}} & \geq\left|n+\nabla_{\theta}.E\left[\hat{\theta}(x)-\theta\right]\right|,\label{eq:GenqCR2-1}
\end{alignat}
where $I_{\beta,q}\left[f|g;\theta\right]$ is the generalized ($\beta,q$)-Fisher
information, and $E_{q}[.]$ denotes the $q$-expectation which is
used in nonextensive statistics. These results are the mutidimensional
extensions, with an arbitrary norm, of our previous $q$-Cramér-Rao
inequalities \cite{bercher_generalized_2012}. In the monodimensional
case and $q=1,$ these inequalities reduce to the Barankin-Vajda Cramér-Rao
inequality, and to the standard Cramér-Rao inequality for $\alpha=\beta=2.$
In addition, in the case of a location parameter, we show that 
\begin{equation}
E\left[\left\Vert x\right\Vert ^{\alpha}\right]^{\frac{1}{\alpha}}\, I_{\beta,q}\left[g\right]^{\frac{1}{\beta}}\geq n\label{eq:GenqCRqGauss}
\end{equation}
which reduces again to our previous results in the univariate case.
Examining carefully the cases for equality in (\ref{eq:GenqCRqGauss}),
we exhibit that the lower bound is attained by generalized $q$-Gaussian
distributions, and we prove that these generalized Gaussian are the
unique extremal functions, provided that the dual norm is strictly
convex. For a random vector $x$ in $\mathbb{R}^{n}$, these generalized
$q$-Gaussian have the probability density 
\begin{equation}
G_{\gamma}(x)=\begin{cases}
\frac{1}{Z(\gamma)}\left(1-\left(q-1\right)\gamma\|x\|^{\alpha}\right)_{+}^{\frac{1}{q-1}} & \text{for }q\not=1\\
\frac{1}{Z(\gamma)}\exp\left(-\gamma\|x\|^{\alpha}\right) & \text{if }q=1
\end{cases}\text{ }\label{eq:qgauss_general-1}
\end{equation}
for $\alpha\in(0,\infty),$ $\gamma$ a real positive parameter and
$q>(n-\alpha)/n,$ where we use the notation $\left(x\right)_{+}=\mbox{max}\left\{ x,0\right\} $,
and where $Z(\gamma)$ is the partition function %
\footnote{ \label{foot} In the case of the Euclidean norm, the general expressions of
the main information measures attached to the generalized Gaussians
are derived in Appendix A of \cite{bercher__2012}. Similar expressions
can readily be obtained in the case of a general norm, using the change
of variable in polar coordinates $x=r\, u,$ with $u=x/||x||$  and
the representation of the Lebesgue measure $\mathrm{d}x=r^{n-1}\mathrm{d}r\,\mathrm{d}\sigma(u)$,
c.f. \cite[p. 87]{stroock_concise_1998}, where $\mathrm{d}\sigma(u)$
denotes the surface element on the unit sphere. By this remark, the
expressions in \cite[Appendix A]{bercher__2012} are valid, with the
proviso that $\omega_{n}$ will denote the volume of the $n$-dimensional
unit ball $\mathcal{B}=\left\{ x\in\mathbb{R}^{n}:\,\|x\|\leq1\right\} $. %
}. For $q>1$, the density has a compact support, while for $q\leq1$
it is defined on the whole $\mathbb{R}^{n}$ and behaves as a power
distribution for $\|x\|\rightarrow\infty.$ A shorthand notation for
the expression of the generalized $q$-Gaussian density is 
\begin{equation}
G_{\gamma}(x)=\frac{1}{Z(\gamma)}\exp_{q^{*}}\left(-\gamma\|x\|^{\alpha}\right),\label{eq:DefQGaussian}
\end{equation}
with $q^{*}=2-q$, and where the so-called $q$-exponential function
is defined by
\begin{equation}
\exp_{q}(x):=\left(1+(1-q)x\right)_{+}^{\frac{1}{1-q}},\text{ for }q\neq1\text{ and }\exp_{q=1}(x):=\exp(x).\label{eq:defExpq-1}
\end{equation}

In section \ref{sec:Another-Cram=0000E9r-Rao-inequality}, we present
another Cramér-Rao type inequality which is also saturated by the
generalized $q$-Gaussian. This inequality has been originally established
by \cite{lutwak_Cramer_2005}, and extended to the multidimensional
case in \cite{lutwak_extensions_2012} and independently in \cite{bercher__2012}
in the case of an Euclidean norm. We show here that this last inequality
can readily be stated and proved in the case of an arbitrary norm. 

Finally, in section \ref{sec:Uncertainty}, we derive some new multidimensional
uncertainty relations from the generalized $q$-Cramér-Rao inequalities.
These uncertainty relations involve escort mean values and are saturated
by generalized Gaussians. In particular, we obtain an inequality of
the form 
\begin{equation}
\left(E_{\frac{k}{2}}\left[\left\Vert x\right\Vert _{2}^{\gamma}\right]\right)^{\frac{1}{\gamma}}\left(E\left[\left\Vert \xi\right\Vert _{2}^{\theta}\right]\right)^{\mbox{\ensuremath{\frac{1}{\theta\lambda}}}}>\frac{1}{M_{\frac{k}{2}}[|\psi|^{2}]^{\frac{1}{k\lambda}}}\,\left(E_{\frac{k}{2}}\left[\left\Vert x\right\Vert _{2}^{\gamma}\right]\right)^{\frac{1}{\gamma}}\left(E\left[\left\Vert \xi\right\Vert _{2}^{\theta}\right]\right)^{\mbox{\ensuremath{\frac{1}{\theta\lambda}}}}\geq K,\label{eq:uncertainty_intro}
\end{equation}
where $x$ and $\xi$ are two Fourier dual variables, $\gamma\geq2,$
$\theta\geq2$, and $E_{\frac{k}{2}}[.]$ denotes an expectation computed
with respect to an escort distribution of order $\frac{k}{2}.$ The
lower bound $K$ is fixed and attained when the underlying wave function
is a $q$-Gaussian. For $\gamma=\theta=2,$ $q=1,$ (\ref{eq:uncertainty_intro})
gives the multidimensional version of the well-known Weyl-Heisenberg
uncertainty principle.

In order to derive these different inequalities, we will need some
preliminary results, in particular concerning some properties of general
norms on $\mathbb{R}^{n}.$ This is the objective of section \ref{sec:Preliminary-results}
where we first define the notion of dual norms and prove a result
on the gradient of a general norm. Then, we establish, together with
its equality conditions, a general Hölder-type inequality. This inequality
will be an essential ingredient in the derivation of the general Cramér-Rao
inequality (\ref{eq:VeryGeneralCR}).

\section{Preliminary results\label{sec:Preliminary-results}}

As mentioned above, we will consider here general norms on $\mathbb{R}^{n}.$
Let us first simply recall that a norm on $\mathbb{R}^{n}$ is a function
$\|.\|$: $\mathbb{R}^{n}\rightarrow\mathbb{R}_{+}$ such that for
any $x$, $y\in\mathbb{R}^{n}$ and $\gamma\in\mathbb{R}$, then 
\begin{equation}
\mbox{(a) \ensuremath{\|\gamma x\|=|\gamma|\|x\|}, (b) \ensuremath{\|x+y\|\leq\|x\|+\|y\|}, and (c) \ensuremath{\|x\|=0} \, iff\,\  \ensuremath{x=0}. }
\end{equation}
A large class of norms is the class of $L_{p}$-norms, $p\geq1,$
given by $\|x\|_{p}=\left(\sum_{i=1}^{n}|x|_{i}^{p}\right)^{\frac{1}{p}}.$
As important particular cases, we have the $L_{1}$-norm, $\|x\|_{1}=\sum_{i=1}^{n}|x|_{i}$,
the max-norm or $L_{\infty}$-norm $\|x\|_{\infty}=\max\,\left(|x_{1}|,\ldots|x_{n}|\right),$
and of course the Euclidean $L_{2}$-norm $\|x\|_{2}=\left(\sum_{i=1}^{n}x_{i}^{2}\right)^{\frac{1}{2}}.$
We shall mention that it is possible to use weighted versions of
the norms above, e.g. $\|x\|_{w,p}=\left(\sum_{i=1}^{n}w_{i}|x|_{i}^{p}\right)^{\frac{1}{p}},$
with $\mbox{ }w_{i}>0$, and that any injective linear transformation
$A$ leads to a new norm, such as $\|x\|_{A}=\|Ax\|.$ Finally, it
is also possible to construct new norms by combining different norms
defined on subvectors of $x.$ 

A related important notion is the notion of dual-norm. Let $E=(\mathbb{R}^{n},\left\Vert .\right\Vert )$
a $n$-dimensional normed space, where $\left\Vert .\right\Vert $
is an arbitrary norm, and let us denote $E^{*}=(\mathbb{R}^{n},\left\Vert .\right\Vert _{*})$
its dual space. For $Y\in E^{*},$ the dual norm $\left\Vert .\right\Vert _{*}$
is defined by
\begin{equation}
\left\Vert Y\right\Vert _{*}=\underset{\left\Vert X\right\Vert \leq1}{\mathrm{sup}}X.Y,\label{eq:DefinitionDualNorm}
\end{equation}
where $X.Y$ is the standard scalar product $X.Y=\sum_{i=1}^{n}X_{i}Y_{i}.$
In particular, it is well known that if $\left\Vert .\right\Vert $
is a $L_{p}$-norm, then $\left\Vert .\right\Vert _{*}$ is the $L_{q}$-norm,
where $p$ and $q$ are Hölder conjugates of each other, i.e. $p^{-1}+q^{-1}=1$,
see e.g. \cite[chapter 5]{horn_matrix_1990}. By a direct consequence
of the definition of the dual norm, we always have 
\begin{equation}
X.Y\leq\left\Vert X\right\Vert \left\Vert Y\right\Vert _{*}.
\end{equation}
Note that when the dot product $X.Y$ is negative, we can always take
the minus of one of the elements to get $\left|X.Y\right|=X.\left(-Y\right)\leq\|X\|\|-Y\|_{*}=\|X\|\|Y\|_{*}.$
Hence, we see that we actually have an extension of Hölder's inequality
for vectors:
\begin{equation}
\left|X.Y\right|\leq\left\Vert X\right\Vert \left\Vert Y\right\Vert _{*}.\label{eq:HolderGenDualNorm}
\end{equation}
Obviously, we recover here the Cauchy-Schwarz inequality if $\left\Vert .\right\Vert =\left\Vert .\right\Vert _{2}$
and the standard Hölder inequality for vectors if $\left\Vert .\right\Vert =\left\Vert .\right\Vert _{p},$
and thus $\left\Vert .\right\Vert _{*}=\left\Vert .\right\Vert _{q}.$

In the following, we will need several facts on the gradient of a
norm. These facts are stated in the next Lemma. 
\begin{lem}
\label{LemmeGradient}Let $\left\Vert .\right\Vert $ be differentiable
at $x\in E$, and denote $x^{*}=\nabla_{x}\left\Vert .\right\Vert (x)\in E^{*}$
the gradient of the norm at $x$. The gradient of $\left\Vert x\right\Vert $
satisfies (a) $x.x^{*}=\left\Vert x\right\Vert $ and (b) $\left\Vert x^{*}\right\Vert _{*}=1.$
Furthermore, when the dual norm \textup{$\left\Vert .\right\Vert _{*}$}
is strictly convex, then the gradient \textup{$x^{*}$} is the unique
vector that satisfies (a) and (b). \end{lem}
\begin{proof}
We begin by equality (a). Let $x$ and $v$ two vectors of $E$ and
$\lambda$ a real parameter. By the triangle inequality, we have $\left\Vert x+\lambda v\right\Vert \leq\left\Vert x\right\Vert +\lambda\left\Vert v\right\Vert $,
so that 
\begin{equation}
\lim_{\lambda\rightarrow0}\frac{\left\Vert x+\lambda v\right\Vert -\left\Vert x\right\Vert }{\lambda}\leq\left\Vert v\right\Vert .\label{eq:deriv}
\end{equation}
In the other hand, the chain rule for derivation $\frac{\mathrm{d}\left\Vert Z\right\Vert }{\mathrm{d}\lambda}=\frac{\mathrm{d}Z}{\mathrm{d}\lambda}.\nabla_{Z}\left\Vert Z\right\Vert ,$
with $Z=x+\lambda v$, gives
\begin{equation}
\left.\frac{\mathrm{d}\left\Vert x+\lambda v\right\Vert }{\mathrm{d}\lambda}\right|_{\lambda=0}=v.\nabla_{x}\left\Vert x\right\Vert \leq\left\Vert v\right\Vert ,\label{eq:deriv2}
\end{equation}
where the right inequality follows from (\ref{eq:deriv}). Of course,
taking $v=x$ in (\ref{eq:deriv}) gives the equality sign, and (\ref{eq:deriv2})
becomes $x.\nabla_{x}\left\Vert x\right\Vert =\left\Vert x\right\Vert $,
that is (a). 

By (\ref{eq:deriv2}), we also get that 
\begin{equation}
\nabla_{x}\left\Vert x\right\Vert .\frac{v}{\left\Vert v\right\Vert }\leq1,
\end{equation}
with equality if $v=x.$ Therefore, in the definition of the dual
norm $\left\Vert \nabla_{x}\left\Vert x\right\Vert \right\Vert _{*}=\sup_{\left\Vert w\right\Vert \leq1}w.\nabla_{x}\left\Vert x\right\Vert $,
the supremum is equal to one and is attained for $w=x/\left\Vert x\right\Vert $.
This proves (b). 

In functional analysis, the existence of a solution to a system analogue
to (a),(b) is granted by a consequence of the Hahn-Banach theorem,
see e.g. the background material in \cite{bonnans_perturbation_2000}.
In this context, the uniqueness of extension in the Hahn-Banach theorem,
therefore the uniqueness of $x^{*},$ is guaranteed if the primal
space $E$ is smooth, which in turn is equivalent to the strict convexity
of the dual norm \cite[Chapter 2]{limaye_functional_1996}. In our
setting, it is easy to check that we have uniqueness of the solution
to the system (a),(b) provided that the dual norm is strictly convex.
Indeed, if $x_{1}^{*}$ and $x_{2}^{*}$ are two solutions to (a),(b),
we have by (a) $\frac{x}{\left\Vert x\right\Vert }.\left(x_{1}^{*}+x_{2}^{*}\right)=2$.
Accordingly, the dual norm $\left\Vert x_{1}^{*}+x_{2}^{*}\right\Vert _{*}=\sup_{\left\Vert w\right\Vert \leq1}w.\left(x_{1}^{*}+x_{2}^{*}\right)$
is necessarily greater than 2: $\left\Vert x_{1}^{*}+x_{2}^{*}\right\Vert _{*}\geq2.$
In the other hand, if the dual norm is strictly convex and using (b),
we have $\left\Vert x_{1}^{*}+x_{2}^{*}\right\Vert _{*}\leq\left\Vert x_{1}^{*}\right\Vert _{*}+\left\Vert x_{2}^{*}\right\Vert _{*}=2$,
with equality if and only if $x_{1}^{*}=x_{2}^{*}$. Combining the
two inequalities, we see that the two solutions are necessarily equal.
Finally, since we have already identified that the gradient of the
norm satisfies (a),(b), we get the last item in the Lemma. 
\end{proof}
The standard Hölder's inequality works for functions and relates the
$L_{1}$ norm of the product of two functions to the product of their
$L_{p}$ and $L_{q}$ norms: $\|fg\|_{1}\le\|f\|_{p}\|g\|_{q}$, with
$1\leq p,\, q\leq\infty$ and $1/p+1/q=1$. For vectors and an arbitrary
norm $\|.\|$, the inequality (\ref{eq:HolderGenDualNorm}) gives
another kind of Hölder inequality (actually, this inequality is also
true in a broader context, see e.g. \cite{morrison_functional_2000}).
By combining these two inequalities, we obtain another Hölder-type
inequality for vector-valued functions, which involves arbitrary norms.
This inequality will be a key in the derivation of the new multidimensional
Cramér-Rao inequality. It is given, with its equality condition, in
the following Lemma.
\begin{lem}
\label{LemmaGenHolderIneq}Let $E=(\mathbb{R}^{n},\left\Vert .\right\Vert )$
be a $n$-dimensional normed space and $E^{*}=(\mathbb{R}^{n},\left\Vert .\right\Vert _{*})$
its dual space. If $X(t)$ and $Y(t)$ are two functions taking values
respectively in $E$ and $E^{*}$, and if $w(t)$ is a weight function,
then 
\begin{alignat}{1}
\left(\int\|X(t)\|^{\alpha}\, w(t)\mathrm{d}t\right)^{\frac{1}{\alpha}}\left(\int\|Y(t)\|_{*}^{\beta}\, w(t)\mathrm{d}t\right)^{\frac{1}{\beta}}\geq & \int\left|X(t).Y(t)\right|\, w(t)\mathrm{d}t,\label{eq:GenHolderInequality-a}\\
\geq & \left|\int X(t).Y(t)\, w(t)\mathrm{d}t\right|\label{eq:GenHolderInequality-b}
\end{alignat}
with $\alpha$ and $\beta$ Hölder conjugates of each other, i.e.
$\alpha^{-1}+\beta^{-1}=1,$ $\alpha\geq1$. The equality is obtained
if 
\begin{equation}
Y(t)=K\|X(t)\|^{\alpha-1}\nabla_{X(t)}\|X(t)\|,\label{eq:ExtremalGenHolderInequality}
\end{equation}
with $K\in\mathbb{R}$ for inequality (\ref{eq:GenHolderInequality-a}),
and with $K\in\mathbb{R}_{+}$ for the lower bound (\ref{eq:GenHolderInequality-b}).
If the dual norm is strictly convex, then the function $Y(t)$ in
(\ref{eq:ExtremalGenHolderInequality}) above is the unique function
which saturates the inequalities (\ref{eq:GenHolderInequality-a}-\ref{eq:GenHolderInequality-b}). \end{lem}
\begin{proof}
By inequality (\ref{eq:HolderGenDualNorm}), we have $\left|X(t).Y(t)\right|\leq\|X(t)\|\|Y(t)\|_{*}.$
Integrating this inequality with respect to $t,$ we obtain 
\begin{equation}
\int\left|X(t).Y(t)\right|\, w(t)\mathrm{d}t\leq\int\|X(t)\|\|Y(t)\|_{*}w(t)\mathrm{d}t.\label{eq:ineq1}
\end{equation}
Obviously, we always have 
\begin{equation}
\left|\int X(t).Y(t)\, w(t)\mathrm{d}t\right|\leq\int\left|X(t).Y(t)\right|\, w(t)\mathrm{d}t,\label{eq:simplepositiveinequality}
\end{equation}
with equality if $X(t).Y(t)\geq0$ everywhere.  Then, it only remains
to apply the standard Hölder inequality to right hand side of (\ref{eq:ineq1}):
\begin{equation}
\int\|X(t)\|\|Y(t)\|_{*}w(t)\mathrm{d}t\leq\left(\int\|X(t)\|^{\alpha}\, w(t)\mathrm{d}t\right)^{\frac{1}{\alpha}}\left(\int\|Y(t)\|_{*}^{\beta}\, w(t)\mathrm{d}t\right)^{\frac{1}{\beta}}\label{eq:ineq2}
\end{equation}
to obtain (\ref{eq:GenHolderInequality-a}). The inequality (\ref{eq:GenHolderInequality-b})
then follows by (\ref{eq:simplepositiveinequality}). 

As far as the cases of equality are concerned, we know that in the
Hölder inequality (\ref{eq:ineq2}), the equality is obtained if and
only if $\|Y(t)\|_{*}^{\beta}=K\|X(t)\|^{\alpha}$, with $K$ a positive
constant. Using the fact that $\alpha/\beta=\alpha-1,$ the condition
becomes $\|Y(t)\|_{*}=K\|X(t)\|^{\alpha-1}$. This condition implies
that $Y(t)$ must be of the form 
\begin{equation}
Y(t)=K\|X(t)\|^{\alpha-1}u,\label{eq:ExpressionY}
\end{equation}
where $u$ is a vector of $E^{*}$ with unit norm: $\|u\|_{*}=1.$
By inequality (\ref{eq:HolderGenDualNorm}) we see that the integrand
in the left side of (\ref{eq:ineq1}) is always less or equal the
integrand on the right. Therefore, we will only get equality in (\ref{eq:ineq1})
if the integrands are equal. Then if we plug (\ref{eq:ExpressionY})
in the inequality (\ref{eq:HolderGenDualNorm}), we obtain the condition
$\|X(t)\|^{\alpha-1}\left|X(t).u\right|=\|X(t)\|^{\alpha}$, that
is finally $\left|X(t).u\right|=\|X(t)\|.$ Since we know by Lemma
\ref{LemmeGradient} that $v=\nabla_{X(t)}\|X(t)\|$ is a unit vector
that satisfies $X(t).v=\|X(t)\|$, and is unique if the dual norm
is strictly convex, we see that $u=\pm v=\pm\nabla_{X(t)}\|X(t)\|$
and this concludes the proof of the first inequality. For equality
to hold in the lower bound (\ref{eq:GenHolderInequality-b}), the
integrand must be positive, which in turn implies that $X(t).u=\|X(t)\|$
and $u=\nabla_{X(t)}\|X(t)\|.$
\end{proof}

\section{The generalized Cramér-Rao inequality\label{sec:The-generalized-Cram=0000E9r-Rao}}

In this section, we first derive a main Cramér-Rao inequality for
the estimation of a multidimensional parameter and introduce a generalized
version of Fisher information. Next, we examine the particular case
of a pair of escort distributions, and then the case of a location
parameter. So doing, we obtain multidimensional versions of the $q$-Cramér-Rao
inequality and a Cramér-Rao inequality characterizing generalized
$q$-Gaussian distributions.

\subsection{The main Cramér-Rao inequality for the estimation of a parameter\label{sub:The-main-Cram=0000E9r-Rao}}

The problem of estimation is to determine a function $\hat{\theta}(x)$
in order to estimate an unknown parameter $\theta.$ Let $f(x;\theta)$
and $g(x;\theta)$ be two probability density functions, with $x\in X\subseteq\mathbb{R}^{k}$
and $\theta$ a parameter of these densities, $\theta\in\mathbb{R}^{n}$.
An underlying idea in the statement of the new Cramér-Rao inequality
is that it is possible to evaluate the moments of the error with respect
to different probability distributions. For instance, in the estimation
setting the estimation error is $\hat{\theta}(x)-\theta$. The bias
can be evaluated with respect to $f$ according to
\begin{equation}
B_{f}(\theta)=\int_{X}\left(\hat{\theta}(x)-\theta\right)\, f(x;\theta)\,\text{d}x=\mathrm{E}_{f}\left[\hat{\theta}(x)-\theta\right]\label{eq:DefinitionBiasVsf}
\end{equation}
 while a general moment of a norm of the error can be computed with
respect to another distribution, $g(x;\theta)$, as in
\begin{equation}
\mathrm{E}_{g}\left[\left\Vert \hat{\theta}(x)-\theta\right\Vert ^{\beta}\right]=\int_{X}\left\Vert \hat{\theta}(x)-\theta\right\Vert ^{\beta}\, g(x;\theta)\,\text{d}x
\end{equation}
 The distributions $f(x;\theta)$ and $g(x,\theta)$ can be chosen
arbitrarily and are not necessarily directly related. However, $g(x;\theta)$
can be designed as a transformation of $f(x;\theta)$ that highlights,
or perhaps scores out, some characteristics of $f(x;\theta)$. Typically,
$g(x;\theta)$ can be a weighted version of $f(x;\theta)$, i.e.
$g(x;\theta)=h(x;\theta)f(x;\theta).$ The distribution $g(x;\theta)$
can also be a quantized version of $f(x;\theta)$, such as $g(x;\theta)=\left[f(x;\theta)\right],$
where $\left[.\right]$ denotes the integer part. Another important
special case is when $g(x;\theta)$ is defined as the escort distribution
of order $q$ of $f(x;\theta)$, where $q$ plays the role of a tuning
parameter. We will see that this special case, which is particularly
important in the context of nonextensive statistical physics, will
lead to generalized $q$-Gaussians as the extremal functions. We are
now in position to state and prove an extended version of the Cramér-Rao
inequality.
\begin{thm}
\label{GenericCRTheorem} Let $f(x;\theta)$ be a multivariate probability
density function defined over a subset $X\mathbb{\subseteq R}^{n}$,
and $\theta\in\Theta\subseteq\mathbb{R}^{k}$ a parameter of the density.
The set $\Theta$ is equipped with a norm $\|.\|,$ and the corresponding
dual norm is denoted $\|.\|_{*}$. Let $g(x;\theta)$ denote another
probability density function also defined on $(X;\Theta)$. Assume
that $f(x;\theta)$ is a jointly measurable function of $x$ and $\theta,$
is integrable with respect to $x$, is absolutely continuous with
respect to $\theta,$ and that the derivatives with respect to each
component of $\theta$ are locally integrable. For any estimator $\hat{\theta}(x)$
of $\theta$, we have 
\begin{gather}
E_g\left[\left\Vert \hat{\theta}(x)-\theta\right\Vert ^{\alpha}\right]^{\frac{1}{\alpha}}I_{\beta}[f|g;\theta]^{\frac{1}{\beta}}\geq\left|n+\nabla_{\theta}.\, B_{f}(\theta)\right|\label{eq:GeneralizedCramerRao-4}
\end{gather}
with $\alpha$ and $\beta$ Hölder conjugates of each other, i.e.
$\alpha^{-1}+\beta^{-1}=1,$ $\alpha\geq1$, and where the $(\beta,g)$-Fisher
information 
\begin{alignat}{1}
I_{\beta}[f|g;\theta] & =\int_{X}\left\Vert \frac{\nabla_{\theta}f(x;\theta)}{g(x;\theta)}\right\Vert _{*}^{\beta}g(x;\theta)\text{\,\ d}x\label{eq:GeneralizedFisher-4}
\end{alignat}
is the generalized Fisher information of order $\beta$ on the parameter
$\theta$ contained in the distribution $f$ and taken with respect
to $g$. The equality case is obtained if
\begin{equation}
\frac{\nabla_{\theta}f(x;\theta)}{g(x;\theta)}=K\left\Vert \hat{\theta}(x)-\theta\right\Vert ^{\alpha-1}\nabla_{\hat{\theta}(x)-\theta}\|\hat{\theta}(x)-\theta\|,\label{eq:CaseOfEqualityInCR-2}
\end{equation}
with $K>0.$\end{thm}
\begin{proof}
The bias in (\ref{eq:DefinitionBiasVsf}) is a $n$-dimensional vector.
Let us consider its divergence with respect to variations of $\theta$:
\begin{equation}
\mathrm{div}\, B_{f}(\theta)=\nabla_{\theta}.\, B_{f}(\theta).
\end{equation}
The regularity conditions in the statement of the theorem enable to
interchange integration with respect to $x$ and differentiation with
respect to $\theta,$ and 
\begin{equation}
\nabla_{\theta}.\, B_{f}(\theta)=\int_{X}\nabla_{\theta}.\left(\hat{\theta}(x)-\theta\right)\, f(x;\theta)\,\text{d}x+\int_{X}\nabla_{\theta}f(x;\theta).\left(\hat{\theta}(x)-\theta\right)\,\text{d}x.
\end{equation}
In the first term on the right, we have $\nabla_{\theta}.\theta=n$,
and the integral reduces to $-n\int_{X}f(x;\theta)\,\text{d}x=-n,$
since $f(x;\theta)$ is a probability density on $X.$ The second
term can be rearranged so as to obtain an integration with respect
to the density $g(x;\theta),$ assuming that the derivatives with
respect to each component of $\theta$ are absolutely continuous with
respect to $g(x;\theta),$ i.e. $g(x;\theta)\gg\nabla_{\theta}f(x;\theta)$.
This gives
\begin{equation}
n+\nabla_{\theta}.\, B_{f}(\theta)=\int_{X}\frac{\nabla_{\theta}f(x;\theta)}{g(x;\theta)}.\left(\hat{\theta}(x)-\theta\right)\, g(x;\theta)\,\text{d}x.
\end{equation}
Now, it only remains to apply the generalized Hölder-type inequality
(\ref{eq:GenHolderInequality-b}) in Lemma \ref{LemmaGenHolderIneq}
to the integral on the right side, with $X(x)=\hat{\theta}(x)-\theta,$
$Y(x)=\frac{\nabla_{\theta}f(x;\theta)}{g(x;\theta)},$ and $w(x)=g(x;\theta).$
This yields in all generality 
\begin{equation}
\left(\int_{X}\left\Vert \hat{\theta}(x)-\theta\right\Vert ^{\alpha}g(x;\theta)\text{\,\ d}x\right)^{\frac{1}{\alpha}}\left(\int_{X}\left\Vert \frac{\nabla_{\theta}f(x;\theta)}{g(x;\theta)}\right\Vert _{*}^{\beta}g(x;\theta)\text{\,\ d}x\right)^{\frac{1}{\beta}}\geq\left|n+\nabla_{\theta}.\, B_{f}(\theta)\right|
\end{equation}
which is (\ref{eq:GeneralizedCramerRao-4}). By Lemma \ref{LemmaGenHolderIneq}
again, we know that the case of equality occurs if $Y(t)=K\|X(t)\|^{\alpha-1}\nabla_{X(t)}\|X(t)\|,$
$K>0,$ which gives (\ref{eq:CaseOfEqualityInCR-2}).
\end{proof}

\subsection{Main consequences of the general result\label{sub:Main-consequences-of}}

\subsubsection{Case of a $q$-escort distribution\label{sub:Case-of-a-q-escort}}

Let $f(x;\theta)$ and $g(x;\theta)$ be a pair of of $q$-escort
distributions linked by 
\begin{equation}
f(x;\theta)=\frac{g(x;\theta)^{q}}{M_{q}\left[g;\theta\right]}\,\,\,\,\text{ and }\,\,\, g(x;\theta)=\frac{f(x;\theta)^{\bar{q}}}{M_{\bar{q}}\left[f;\theta\right]},\label{eq:PairEscorts}
\end{equation}
with $q>0$, $\bar{q}=1/q$, and the information generating function
$M_{q}\left[g;\theta\right]$ defined by 
\begin{equation}
M_{q}\left[g;\theta\right]=\int_{X}g(x;\theta)\,\text{d}x.
\end{equation}
As usual, we will denote by $E_{q}\left[.\right]$ the $q$-expectation,
which is the expectation taken with respect to an escort distribution
of order $q$. Here we see that the expectation with respect to $f(x;\theta)$
is also the $q$-expectation with respect to $g(x;\theta).$ Let us
also recall that the inverse function of the deformed $q$-exponential
(\ref{eq:defExpq-1}), the so-called $q$-logarithm, is defined by
\begin{equation}
\ln_{q}(x):=\frac{x^{1-q}-1}{1-q}.\label{eq:defLnq-1-1}
\end{equation}
With these notations, we have the following corollary of the general
Cramér-Rao inequality.
\begin{cor}
\label{Corollary1}For the pair of escort distributions (\ref{eq:PairEscorts}),
the equivalent Cramér-Rao inequalities 
\begin{alignat}{1}
E\left[\left\Vert \hat{\theta}(x)-\theta\right\Vert ^{\alpha}\right]^{\frac{1}{\alpha}}\, I_{\beta,q}\left[f|g;\theta\right]^{\frac{1}{\beta}} & \geq\left|n+\nabla_{\theta}.E_{q}\left[\hat{\theta}(x)-\theta\right]\right|\label{eq:GenqCR1}\\
E_{\bar{q}}\left[\left\Vert \hat{\theta}(x)-\theta\right\Vert ^{\alpha}\right]^{\frac{1}{\alpha}}\, I_{\beta,q}\left[f|g;\theta\right]^{\mbox{\ensuremath{\frac{1}{\beta}}}} & \geq\left|n+\nabla_{\theta}.E\left[\hat{\theta}(x)-\theta\right]\right|,\label{eq:GenqCR2}
\end{alignat}
 hold, where the generalized ($\beta,q$)-Fisher information is given
by
\begin{align}
I_{\beta,q}\left[f|g;\theta\right] & =\frac{1}{M_{q}\left[g;\theta\right]^{\beta}}\, E\left[g(x;\theta)^{\beta(q-1)}\left\Vert \nabla_{\theta}\ln\frac{g(x;\theta)^{q}}{M_{q}\left[g;\theta\right]}\right\Vert _{*}^{\beta}\right]\label{eq:GenqFish1}\\
 & =M_{\bar{q}}\left[f;\theta\right]^{\beta}\,\, E_{\bar{q}}\left[f(x;\theta)^{\beta(1-\bar{q})}\left\Vert \nabla_{\theta}\ln f(x;\theta)\right\Vert _{*}^{\beta}\right],\label{eq:GenqFish2}
\end{align}
and where equality occurs if
\begin{equation}
\nabla_{\theta}\ln_{\bar{q}}f(x;\theta)=K\left\Vert \hat{\theta}(x)-\theta\right\Vert ^{\alpha-1}\nabla_{\hat{\theta}(x)-\theta}\|\hat{\theta}(x)-\theta\|,\label{eq:EqualityInqCR}
\end{equation}
with $K>0.$\end{cor}
\begin{proof}
The Cramér-Rao inequalities (\ref{eq:GenqCR1}) and (\ref{eq:GenqCR2})
directly follow from the general Cramér-Rao inequality (\ref{eq:GeneralizedCramerRao-4}),
the relations (\ref{eq:PairEscorts}) between $f(x;\theta)$ and $g(x;\theta)$,
and the notation of $q$-expectations. The expressions of the generalized
($\beta,q$)-Fisher information also follow by direct calculation.
Finally, the equality condition yields 
\begin{equation}
f(x;\theta)^{(1-\bar{q})}\nabla_{\theta}\ln f(x;\theta)=K\left\Vert \hat{\theta}(x)-\theta\right\Vert ^{\alpha-1}\nabla_{\hat{\theta}(x)-\theta}\|\hat{\theta}(x)-\theta\|.
\end{equation}
Noticing that the term on the left is nothing but the gradient of
the deformed $q$-logarithm, $f(x;\theta)^{(1-\bar{q})}\nabla_{\theta}\ln f(x;\theta)=\nabla_{\theta}\ln_{\bar{q}}\, f(x;\theta)$,
we immediately obtain (\ref{eq:EqualityInqCR}).
\end{proof}

\subsubsection{Case of a translation family\label{sub:Case-of-a-translation}}

In the particular case of a translation parameter, the generalized
Cramér-Rao inequality induces a new class of inequalities. 

Let $\theta\in\mathbb{R}^{n}$ be a location parameter, $x\in X\subseteq\mathbb{R}^{n}$,
and define by $f(x;\theta)$ the family of density $f(x;\theta)=f(x-\theta)$.
In this case, we have $\nabla_{\theta}f(x;\theta)=-\nabla_{x}f(x-\theta),$
provided that $f$ is differentiable at $x-\theta$, and the Fisher
information becomes a characteristic of the information in the distribution.
If $X$ is a bounded subset, we will assume that $f(x)$ vanishes
and is differentiable on the boundary $\partial X$ (otherwise the
Fisher information defined for the function extended to $\mathbb{R}^{n}$
is not defined). 

Let us denote by $\mu_{f}$ the mean of $f(x).$ We immediately get
that the mean of $f(x;\theta)$ is $(\mu_{f}+\theta)$, so that an
unbiased estimator of $\theta$ could be $\hat{\theta}(x)=x-\mu_{f}$.
If we choose $\hat{\theta}(x)=x$, the estimator will be biased, $B_{f}(\theta)=$$E_{q}\left[\hat{\theta}(x)-\theta\right]=\mu_{f}$,
but independent of $\theta$, so that the gradient of the bias with
respect to $\theta$ is zero. In these conditions, the generalized
Cramér-Rao inequality becomes 
\begin{equation}
\left(\int\left\Vert x-\theta\right\Vert ^{\alpha}g(x;\theta)\text{\,\ d}x\right)^{\frac{1}{\alpha}}\left(\int\left\Vert \frac{\nabla_{x}f(x-\theta)}{g(x;\theta)}\right\Vert _{*}^{\beta}g(x;\theta)\text{\,\ d}x\right)^{\frac{1}{\beta}}\geq n.
\end{equation}
Furthermore, we can also choose $\theta=0,$ and obtain, as a corollary,
the following interesting functional inequality. 
\begin{cor}
\label{Corollary2}Let $f(x)$ and $g(x)$ be two multivariate probability
density functions defined over a subset $X$ of $\mathbb{R}^{n}$.
Assume that $f(x)$ is a measurable differentiable function of $x$,
which vanishes and is differentiable on the boundary $\partial X$,
that $\nabla_{x}f(x)$ is absolutely continuous with respect to $g(x)$,
and finally that the involved integrals exist and are finite. Then,
the following inequality holds
\begin{equation}
\left(\int_{X}\left\Vert x\right\Vert ^{\alpha}g(x)\text{\,\ d}x\right)^{\frac{1}{\alpha}}\left(\int_{X}\left\Vert \frac{\nabla_{x}f(x)}{g(x)}\right\Vert _{*}^{\beta}g(x)\text{\,\ d}x\right)^{\frac{1}{\beta}}\geq n,\label{eq:CRInequalityLocation}
\end{equation}
with equality if (and only if when the dual norm is strictly convex)
\begin{equation}
\nabla_{x}f(x)=-K\, g(x)\|x\|^{\alpha-1}\nabla_{x}\|x\|,\label{eq:EqualityInCRInequalityLocation}
\end{equation}
with $K>0.$ 
\end{cor}
As an elementary application, let us consider the univariate case,
with $X=[0,1].$ Let us take for $g(x)$ the uniform distribution
on the interval. Finally, let us choose for $f(x)$ a $\beta$-distribution:
$f(x)=x^{a-1}(1-x)^{b-1}/B(a,b),$ with $B(a,b)$ the $\beta$-function.
Firstly, we obviously have $\int_{0}^{1}|x|^{\alpha}\text{d}x=1.$
Secondly, $f'(x)=\left((a-1)x^{a-2}(1-x)^{b-1}+(b-1)x^{a-1}(1-x)^{b-2}\right),$
so that the inequality is 
\begin{equation}
\left(\int_{0}^{1}\left|(a-1)x^{a-2}(1-x)^{b-1}+(b-1)x^{a-1}(1-x)^{b-2}\right|{}^{\beta}\text{\,\ d}x\right)^{\frac{1}{\beta}}\geq B(a,b).
\end{equation}
Taking now $\beta=1$ and $a>1,\, b>1$, we obtain the following
inequality for $\beta$-functions: 
\begin{equation}
(a-1)B(a-1,b)+(b-1)B(a,b-1)\geq B(a,b)
\end{equation}

\subsubsection{Case of a location parameter within a pair of escort distributions\label{sub:Case-of-both}}

By combining the two aspects presented above, namely the case of a
pair of escort distributions and the case of a location parameter,
we will obtain a new Cramér-Rao inequality saturated by multivariate
generalized $q$-Gaussians. This provides a new information theoretic
characterization of generalized $q$-Gaussian and extend our previous
results to the multivariate case and arbitrary norms. As in Corollary\,\ref{Corollary1},
we use a pair of of $q$-escort distributions: 
\begin{equation}
f(x)=\frac{g(x)^{q}}{M_{q}\left[g\right]}\,\,\,\,\text{ and }\,\,\, g(x)=\frac{f(x)^{\bar{q}}}{M_{\bar{q}}\left[f\right]},\label{eq:PairEscorts-1}
\end{equation}
with $\bar{q}=1/q,$ and we denote by $E\left[.\right]$ the standard
expectation with respect to $g(x)$, and by $E_{\bar{q}}\left[.\right]$
the $\bar{q}$-expectation with respect to $f(x)$, which is simply
the standard expectation taken with respect to the escort $f(x)^{\bar{q}}/M_{\bar{q}}\left[f\right].$
In the statement of the following corollary, we will use the deformed
exponential and logarithm defined in (\ref{eq:defExpq-1}),(\ref{eq:defLnq-1-1}).
We will also use the notation $q_{*}=2-q$ that changes the quantities
$(1-q_{*})$ into $(q-1)$. 
\begin{cor}
Let $g(x)$ be a multivariate probability density function defined
over a subset $X\subseteq\mathbb{R}^{n}$. Assume that $g(x)$ is
a measurable differentiable function of $x$, which vanishes and is
differentiable on the boundary $\partial X$, and finally that the
involved integrals exist and are finite. Then, for the pair of escort
distributions (\ref{eq:PairEscorts-1}), the following $q$-Cramér-Rao
inequality holds
\begin{equation}
m_{\alpha}[g]^{\frac{1}{\alpha}}\, I_{\beta,q}\left[g\right]^{\frac{1}{\beta}}\geq n\label{eq:ExtendedqCRloc1}
\end{equation}
with
\begin{equation}
\begin{cases}
\begin{array}{l}
m_{\alpha}[g]=E\left[\left\Vert x\right\Vert ^{\alpha}\right]\\
\begin{alignedat}[t]{1}I_{\beta,q}\left[g\right] & =\left(q/M_{q}\left[g\right]\right)^{\beta}\, E\left[g(x)^{\beta(q-1)}\left\Vert \nabla_{x}\ln g(x)\right\Vert _{*}^{\beta}\right]\\
 & =\left(q/M_{q}\left[g\right]\right)^{\beta}\, E\left[\left\Vert \nabla_{x}\ln_{q*}g(x)\right\Vert _{*}^{\beta}\right],
\end{alignedat}
\end{array}\end{cases}
\end{equation}
where $\alpha$ and $\beta$ are Hölder conjugates of each other,
i.e. $\alpha^{-1}+\beta^{-1}=1,$ $\alpha\geq1$, and where \textup{$I_{\beta,q}\left[g\right]$}
denotes the generalized ($\beta,q$)-Fisher information.

In terms of $\bar{q}$-expectations with respect to $f(x),$ it can
also be written

\begin{equation}
m_{\alpha,\bar{q}}[f]^{\frac{1}{\alpha}}\,\bar{I}_{\beta,q}\left[f\right]^{\frac{1}{\beta}}\geq n\label{eq:ExtendedqCRloc2}
\end{equation}
with
\begin{equation}
\begin{cases}
\begin{array}{l}
m_{\alpha,\bar{q}}[f]=E_{\bar{q}}\left[\left\Vert x\right\Vert ^{\alpha}\right]\\
\begin{alignedat}[t]{1}\bar{I}_{\beta,q}\left[f\right] & =M_{\bar{q}}\left[f\right]^{\beta}\,\, E_{\bar{q}}\left[f(x)^{\beta(1-\bar{q})}\left\Vert \nabla_{x}\ln f(x)\right\Vert _{*}^{\beta}\right]\\
 & =M_{\bar{q}}\left[f\right]^{\beta}\,\, E_{\bar{q}}\left[\left\Vert \nabla_{x}\ln_{\bar{q}}f(x)\right\Vert _{*}^{\beta}\right].
\end{alignedat}
\end{array}\end{cases}
\end{equation}

In both cases, equality occurs if
\begin{equation}
f(x)\propto\exp_{\bar{q}}\left(-\gamma\left\Vert x\right\Vert ^{\alpha}\right),\,\,\text{ or equivalently }\,\, g(x)\propto\exp_{q_{*}}\left(-\gamma\left\Vert x\right\Vert ^{\alpha}\right),\mathrm{\,\, with}\,\,\gamma>0.\label{eq:EqualityInqCR-1}
\end{equation}
If the dual norm is strictly convex, then this generalized $q$-Gaussian
is the unique probability density function that achieves the equality
in the extended Cramér-Rao inequalities. \end{cor}
\begin{proof}
As indicated above, the result is a direct consequence of Theorem\,\ref{GenericCRTheorem}
in the case of a pair of escort distributions and of the estimation
of a location parameter, with $\hat{\theta}(x)=x$ and $\theta=0.$
Using (\ref{eq:PairEscorts-1}), the condition for equality (\ref{eq:EqualityInCRInequalityLocation})
becomes 
\begin{equation}
g(x)^{q-1}\nabla_{x}g(x)=-K\,\|x\|^{\alpha-1}\nabla_{x}\|x\|\, g(x).
\end{equation}
From this equation, we see that $g(x)$ will only be a function of
the norm of $x$, and therefore will be radially symmetric. Furthermore,
we see that the gradient of $g(x)$ behaves as the negative of the
gradient of $\|x\|.$ This means that $g(x),$ which is a function
of $\|x\|,$ is non increasing with $\|x\|.$  

For $g(x)\neq0,$ the equality condition can be written 
\begin{equation}
g(x)^{q-2}\nabla_{x}g(x)=\frac{1}{q-1}\nabla_{x}g(x)^{q-1}=-\frac{K}{\alpha}\,\nabla_{x}\|x\|^{\alpha},
\end{equation}
which, after integration of the two sides, gives 
\begin{equation}
g(x)^{q-1}=-\frac{K}{\alpha}(q-1)\|x\|^{\alpha}+C,
\end{equation}
where $C$ is a constant of integration. Since $g(x)$ is a probability
density function, the solution is restricted to the domain where the
right hand side is non negative, and $g(x)=0$ elsewhere. In particular,
when $C$ is negative, we see that $g(x)$ vanishes around the origin
and presents a singularity at $\|x\|=C\alpha/K(q-1)$. Since we assumed
$g(x)$ differentiable everywhere, this solution must be discarded. 

Therefore, the constant of integration $C$ must be positive, and
\begin{equation}
g(x)\propto\left(1-\frac{K}{C\alpha}(q-1)\|x\|^{\alpha}\right)^{\frac{1}{q-1}}\propto\exp_{q_{*}}\left(-\gamma\left\Vert x\right\Vert ^{\alpha}\right)
\end{equation}
which is (\ref{eq:EqualityInqCR-1}). The expression of $f(x)$ simply
follows from the fact that $f(x)$ is the escort distribution of $g(x).$
\end{proof}

\section{Another Cramér-Rao inequality saturated by generalized Gaussians\label{sec:Another-Cram=0000E9r-Rao-inequality}}

We finish this paper with another Cramér-Rao type inequality, which
involves a variant $\phi_{\beta,q}\left[g\right]$ of the generalized
Fisher information $I_{\beta,q}\left[g\right]$ above, and which is
also saturated by generalized $q$-Gaussian distributions. However,
this inequality is less directly related to estimation results than
the inequality (\ref{eq:ExtendedqCRloc1}) which is just a special
case of the general Cramér-Rao inequality. The monodimensional version
of this inequality has been established by \cite{lutwak_Cramer_2005},
and extended to the multidimensional case in \cite{lutwak_extensions_2012}
and in \cite{bercher__2012} in the case of an Euclidean norm. Actually,
the inequality can readily be stated and proved for a general norm. 
\begin{thm}
\label{theo2} For $n\geq1,$ $\beta$ and $\alpha$ Hölder conjugates
of each other, $\alpha>1,$ $q>\max\left\{ (n-1)/n,\, n/(n+\alpha)\right\} $
then for any probability density $g$ on $\mathbb{R}^{n}$, supposed
continuously differentiable and such that the involved information
measures are finite,
\begin{equation}
m_{\alpha}\left[g\right]^{\frac{1}{\alpha}}\,\phi_{\beta,q}\left[g\right]^{\frac{1}{\beta\lambda}}\geq m_{\alpha}\left[G\right]^{\frac{1}{\alpha}}\,\phi_{\beta,q}\left[G\right]^{\frac{1}{\beta\lambda}}\label{eq:GeneralizedCramerLutwak}
\end{equation}
with $\lambda=n(q-1)+1$ and where the general Fisher information
is given by 
\begin{equation}
\phi_{\beta,q}\left[g\right]=\left(M_{q}\left[g\right]/q\right)^{\beta}I_{\beta,q}\left[g\right]=\, E\left[g(x)^{\beta(q-1)}\left\Vert \nabla_{x}\ln g(x)\right\Vert _{*}^{\beta}\right],
\end{equation}
and where the equality holds iff $g$ is a generalized Gaussian $g=G_{\gamma}.$
\end{thm}

For the proof of this inequality, we will use two general inequalities
relating the moment $m_{\alpha}\left[g\right]$, the generalized Fisher
information $\phi_{\beta,q}\left[g\right]$ and the information generating
function $M_{q}[g]=\int g(x)^{q}\text{d}x.$ We will also use the
notation $N_{q}[g]=M_{q}[g]^{\frac{1}{1-q}},$which is known as the
``Rényi entropy power''. In the theorem above as well as in the lemma below, $\phi_{\beta,q}\left[G\right]$, 
$m_{\alpha}\left[G\right]$  and $N_{q}[G]$  are the values taken by the Fisher information, the moment of order $\alpha$ and the entropy power when the probability density $g$ is the generalized $q$-Gaussian $G$. 
The exact expressions of these quantities are given in the Appendix A of \cite{bercher__2012}, taking into account the  footnote page \pageref{foot}.
\begin{lem}
For $n\geq1,$ $\alpha\in(0,\infty),$ $q>n/(n+\alpha),$ and if $g$
is a probability density on random vectors of $\mathbb{R}^{n}$ with
\textup{$m_{\alpha}[g]=E[\|x\|^{\alpha}]<\infty,\, N_{q}[g]<\infty$,}
then 
\begin{equation}
\frac{m_{\alpha}[g]^{\frac{1}{\alpha}}}{N_{q}[g]^{\frac{1}{n}}}\geq\frac{m_{\alpha}[G]^{\frac{1}{\alpha}}}{N_{q}[G]^{\frac{1}{n}}},\label{eq:momententropy}
\end{equation}
with equality if and only if $g$ is a generalized Gaussian. \end{lem}
\begin{proof}
The inequality (\ref{eq:momententropy}) has been stated and proved
in \cite{lutwak_moment-entropy_2007} in the case of an Euclidean
norm. We simply note here that the proof in \cite{lutwak_moment-entropy_2007}
works as well in the case of a general norm. 
\end{proof}
We will also use a generalized Stam inequality derived from a general
sharp Gagliardo-Nirenberg inequality proved in the remarkable paper
of Cordero et al. \cite{cordero-erausquin_mass-transportation_2004}. 
\begin{lem}
For $n\geq1,$ $\beta$ and $\alpha$ Hölder conjugates of each other,
$\alpha>1,$ and $q>\max\left\{ (n-1)/n,\, n/(n+\alpha)\right\} $,
then for any probability density on $\mathbb{R}^{n}$, supposed continuously
differentiable, the following generalized Stam inequality holds 
\begin{equation}
\phi_{\beta,q}\left[g\right]^{\frac{1}{\beta\lambda}}\, N_{q}[g]^{\frac{1}{n}}\geq\phi_{\beta,q}\left[G\right]^{\frac{1}{\beta\lambda}}\, N_{q}[G]^{\frac{1}{n}},\label{eq:GeneralizedStamInequality}
\end{equation}
with $\lambda=n(q-1)+1$ and with equality if and only if $g$ is
any generalized $q$-Gaussian (\ref{eq:DefQGaussian}).\end{lem}
\begin{proof}
In our notations, the sharp Gagliardo-Nirenberg inequality \cite[Eq. (34) p. 320]{cordero-erausquin_mass-transportation_2004},
with $a>1$ and $\|\nabla u\|_{\beta}=\left(\int\|\nabla u\|_{*}^{\beta}\text{d}x\right)^{\frac{1}{\beta}}$,
is 
\begin{equation}
\frac{\|\nabla u\|_{\beta}\|u\|_{a(\beta-1)+1}^{\frac{1}{\theta}-1}}{\|u\|_{a\beta}^{\frac{1}{\theta}}}\geq K\label{eq:sharpCordero}
\end{equation}
where $K$ is a sharp constant which is attained if and only if $u$
is a generalized Gaussian with exponent $1/(1-a),$ and where $\theta$
is given by $\theta=n(a-1)/a(n\beta-(a\beta+1-a)(n-\beta))$. The
idea is to take $u=g^{t},$ $g$ being a probability density function,
with $a\beta t=1,$ and to note $q=\left[a(\beta-1)+1\right]t$. With
these notations, we get that $\beta t=\beta(q-1)+1,$ and (\ref{eq:sharpCordero})
becomes
\begin{equation}
\phi_{\beta,q}\left[g\right]^{\frac{1}{\beta}}\, N_{q}[g]^{\frac{(\theta-1)(1-q)t}{\theta q}}\geq K.\label{eq:phiNexpo}
\end{equation}
Simplifying the expression of $\theta$ and the exponent in (\ref{eq:phiNexpo}),
we finally obtain, with $q<1$, the generalized Stam inequality (\ref{eq:GeneralizedStamInequality}),
with equality if and only if $g$ is any generalized $q$-Gaussian
(\ref{eq:DefQGaussian}). Actually, this generalized Stam inequality
is also valid in the case $q>1,$ as it can be checked from \cite{cordero-erausquin_mass-transportation_2004}'s
results using similar steps as above. The conditions on $q$ simply
ensure the existence of the information measures for the generalized
Gaussian. 
\end{proof}
We end with the proof of theorem\,\ref{theo2}, which is now an easy
task. 
\begin{proof}
The Cramér-Rao inequality (\ref{eq:ExtendedqCRloc1}) can also be
written
\begin{equation}
m_{\alpha}\left[g\right]^{\frac{1}{\alpha}}\,\frac{\phi_{\beta,q}\left[g\right]^{\frac{1}{\beta}}}{\frac{n}{q}M_{q}[g]}\geq1.\label{eq:ExtendedqCRLoc1bis}
\end{equation}
Eliminating $M_{q}[g]$ between this inequality and the moment-entropy
inequality (\ref{eq:momententropy}) with $q>1,$ we arrive at (\ref{eq:GeneralizedCramerLutwak}).
Similarly, in the case $q<1,$ the elimination of $M_{q}[g]$ between
the extended $q$-Cramér-Rao inequality for a location parameter (\ref{eq:ExtendedqCRLoc1bis})
and the generalized Stam inequality (\ref{eq:GeneralizedStamInequality})
also leads to (\ref{eq:GeneralizedCramerLutwak}). The case of equality
directly follows from the cases of equality in the initial inequalities.
Alternatively, we can observe that (\ref{eq:GeneralizedCramerLutwak})
also follows at once by the combination of (\ref{eq:momententropy})
and (\ref{eq:GeneralizedStamInequality}). 
\end{proof}

\section{\label{sec:Uncertainty}Uncertainty relations associated with the
$q$-Cram\'er-Rao inequalities}
\enlargethispage{0.5cm}
It is known that the Weyl-Heisenberg uncertainty principle in statistical
physics corresponds to the standard Cram\'er-Rao inequality for the
location parameter, see e.g. \cite{stam_inequalities_1959}. Following
this idea, we derive new uncertainty relations from the extended Cram\'er-Rao
inequalities. These uncertainty relations involve escort mean values
and are saturated by generalized Gaussians. These uncertainty relations
thus give a way to measure the uncertainty not only with respect to
the original wave function, but also with respect to an escort-deformed
version of the associated probability density. Of course, the standard
Weyl-Heisenberg uncertainty inequality is recovered in the case $q=1$. 

Let us recall that the uncertainty principle originates in Heisenberg's
work and is stated in Weyl's book \cite[page 77]{weyl_theory_1950},
who credits Pauli. It indicates that if $x$ and $\xi$ are dual Fourier
variables (e.g. position and momentum), then the less the uncertainty
in $x,$ the greater in $\xi$ and conversely: 
\begin{equation}
E[|x|^{2}]E[|\xi|^{2}]\geq\frac{1}{16\pi^{2}},\label{eq:HeisenbergInequality}
\end{equation}
with equality if and only if the probability density of $x$ is a
Gaussian density. Many improvements and variations on this inequality
have been given, see for instance the review \cite{folland_uncertainty_1997}.
Interesting, but not well-known, moment inequalities have been given
in \cite{angulo_uncertainty_1993,angulo_information_1994}. Similar
formulations are obtained below. Other entropic uncertainty relations
have been stated by Hirshman \cite{hirschman_note_1957} and improved
in \cite{bialynicki-birula_formulation_2006,zozor_classes_2007}.
It is worth mentioning that all the uncertainty inequalities are extremely
useful in the analysis of complex systems, see the recent book \cite{sen_statistical_2011},
particularly the chapter \cite{sen_statistical_2011}, and the paper
\cite{zozor_position-momentum_2011}. 

In the following, we consider a complex amplitude wave function $\psi(x)$,
$x\in\mathbb{R}^{n}$, with unit Euclidean norm 2, and denote $\hat{\psi}(\xi)$
its Fourier transform, $\xi\in\mathbb{R}^{n}$. By the Parseval equality
$\int_{\mathbb{R}^n}|\psi(x)|^{2}\mathrm{d}x=\int_{\mathbb{R}^n}|\hat{\psi}(\xi)|^{2}\mathrm{d}\xi=1$
and both $|\psi(x)|^{2}$ and $|\hat{\psi}(\xi)|^{2}$ are probability
density functions. 

We begin by a simple change of function which enables to express the
generalized Fisher information as a Dirichlet energy. This leads to
the following result. 
\begin{prop}
For $k=\beta/\left(\beta(q-1)+1\right),$ $\lambda=n(q-1)+1$, $q>\max\left\{ (n-1)/n,\, n/(n+\alpha)\right\} $,
we have
\begin{align}
 & \frac{1}{M_{\frac{kq}{2}}[|\psi|^{2}]}\left(\int_{\mathbb{R}^n}\left\Vert x\right\Vert ^{\alpha}|\psi(x)|^{k}\mathrm{d}x\right)^{\frac{1}{\alpha}}\left(\int_{\mathbb{R}^n}\left\Vert \nabla_{x}|\psi(x)|\right\Vert _{*}^{\beta}\mathrm{d}x\right)^{\frac{1}{\beta}}\geq\frac{n}{kq}\label{eq:FirstCRIntermsofPsi}\\
 & \frac{1}{M_{\frac{k}{2}}[|\psi|^{2}]^{\frac{1}{k\lambda}}}\,\left(E_{\frac{k}{2}}\left[\left\Vert x\right\Vert ^{\alpha}\right]\right)^{\frac{1}{\alpha}}\left(\int_{\mathbb{R}^n}\left\Vert \nabla_{x}|\psi(x)|\right\Vert _{*}^{\beta}\mathrm{d}x\right)^{\frac{1}{\beta\lambda}}\geq\frac{1}{k^{\frac{1}{\lambda}}}m_{\alpha}\left[G\right]^{\frac{1}{\alpha}}\,\phi_{\beta,q}\left[G\right]^{\frac{1}{\beta\lambda}},\label{eq:SecondCRIntermsofPsi}
\end{align}
where $E_{q}[.]$ denote the $q$-expectation, the expectation computed
with respect to the escort distribution of order $q$, and $M_{q}[.]$
the information generating function of order $q$. In these inequalities,
the case of equality is obtained if and only if $|\psi(x)|^{k}/M_{\frac{k}{2}}[|\psi|^{2}]$
is a generalized $q$-Gaussian. \end{prop}
\begin{proof}
Let 
\begin{equation}
f(x)=\frac{\left(|\psi(x)|^{2}\right)^{\frac{k}{2}}}{\int_{\mathbb{R}^n}\left(|\psi(x)|^{2}\right)^{\frac{k}{2}}\mathrm{d}x}=\frac{|\psi(x)|^{k}}{M_{\frac{k}{2}}[|\psi|^{2}]},\label{eq:EscortPsitof}
\end{equation}
with $k=\beta/\left(\beta(q-1)+1\right).$ Note that we also have
the conjugation relation $\frac{1}{\alpha}+\frac{1}{k}=q.$ The change
of function (\ref{eq:EscortPsitof}) is chosen in order to reduce
$\phi_{\beta,q}\left[f\right]$ to the simple form (\ref{eq:ChangeFunctionPhi})
given below. In the other hand, the moment of order $\alpha$ with
respect to $f(x)$ is nothing but the generalized $\frac{k}{2}$-escort
moment (\ref{eq:ChangeFunctionM}) of order $\alpha$ computed with
respect to $|\psi(x)|^{2}$. Finally, the information generating function
$M_{q}[f]$ can also be expressed in term of $|\psi(x)|^{2}$: 
\begin{alignat}{1}
 & \phi_{\beta,q}\left[f\right]=\int_{\mathbb{R}^n}f(x)^{\beta(q-1)+1}\left\Vert \nabla_{x}\ln f(x)\right\Vert _{*}^{\beta}\mathrm{d}x=\frac{|k|^{\beta}}{M_{\frac{k}{2}}[|\psi|^{2}]^{\frac{\beta}{k}}}\int_{\mathbb{R}^n}\left\Vert \nabla_{x}|\psi(x,t)|\right\Vert _{*}^{\beta}\mathrm{d}x.\label{eq:ChangeFunctionPhi}\\
 & m_{\alpha}\left[f\right]=\int_{\mathbb{R}^n}\left\Vert x\right\Vert ^{\alpha}\frac{\left(|\psi(x,t)|^{2}\right)^{\frac{k}{2}}}{\int_{\mathbb{R}^n}\left(|\psi(x,t)|^{2}\right)^{\frac{k}{2}}\mathrm{d}x}\mathrm{d}x=E_{\frac{k}{2}}\left[\left\Vert x\right\Vert ^{\alpha}\right].\label{eq:ChangeFunctionM}\\
 & M_{q}[f]=\frac{M_{\frac{kq}{2}}[|\psi|^{2}]}{M_{\frac{k}{2}}[|\psi|^{2}]^{q}}.
\end{alignat}
 With these notations, the relation between $\phi_{\beta,q}\left[f\right]$
and $I_{\beta,q}\left[f\right]$, i.e. $I_{\beta,q}\left[f\right]^{\frac{1}{\beta}}=\frac{\phi_{\beta,q}\left[f\right]^{\frac{1}{\beta}}}{\frac{1}{q}M_{q}[f]}$,
and the conjugation relation $\frac{1}{\alpha}+\frac{1}{k}=q$, the
generalized Cram\'er-Rao (\ref{eq:ExtendedqCRloc1}) becomes (\ref{eq:FirstCRIntermsofPsi}).
Similarly, the second Cram\'er-Rao inequality (\ref{eq:GeneralizedCramerLutwak}) gives (\ref{eq:SecondCRIntermsofPsi}).
\end{proof}

Beginning with the relations (\ref{eq:FirstCRIntermsofPsi}), (\ref{eq:SecondCRIntermsofPsi})
and using a relation between the norm of a function and the norm of
its Fourier transform, we obtain a pair of general uncertainty relations,
for any exponent, and that involve the expectations computed with
respect to escort distributions. This result is stated as follows:
\begin{thm}
For $2\geq\alpha\geq1$, $\beta$ its H\"older conjugate, $q>\max\left\{ (n-1)/n,\, n/(n+\alpha)\right\} ,$
$k=\beta/\left(\beta(q-1)+1\right),$ and $\lambda=n(q-1)+1$, then
we have
\begin{alignat}{1}
 & \frac{M_{\frac{\alpha}{2}}[|\hat{\psi}|^{2}]^{\frac{1}{\alpha}}M_{\frac{k}{2}}[|\psi|^{2}]^{\frac{1}{\alpha}}}{M_{\frac{kq}{2}}[|\psi|^{2}]}\,\, E_{\frac{k}{2}}\left[\left\Vert x\right\Vert _{\alpha}^{\alpha}\right]^{\frac{1}{\alpha}}E_{\frac{\alpha}{2}}\left[\left\Vert \xi\right\Vert _{\beta}^{\beta}\right]^{\frac{1}{\beta}}\geq\frac{n}{2\pi kq}\left(\frac{\beta^{\frac{1}{\beta}}}{\alpha^{\frac{1}{\alpha}}}\right)^{-\frac{n}{2\beta}}.\label{eq:FirstCRuncertainty3}\\
 & \frac{M_{\frac{\alpha}{2}}[|\hat{\psi}|^{2}]^{\frac{1}{\alpha\lambda}}}{M_{\frac{k}{2}}[|\psi|^{2}]^{\frac{1}{k\lambda}}}\,\left(E_{\frac{k}{2}}\left[\left\Vert x\right\Vert _{\alpha}^{\alpha}\right]\right)^{\frac{1}{\alpha}}\left(E_{\frac{\alpha}{2}}\left[\left\Vert \xi\right\Vert _{\beta}^{\beta}\right]\right)^{\mbox{\ensuremath{\frac{1}{\beta\lambda}}}}\geq\frac{1}{(2\pi k)^{\frac{1}{\lambda}}}\left(\frac{\beta^{\frac{1}{\beta}}}{\alpha^{\frac{1}{\alpha}}}\right)^{-\frac{n}{2\beta\lambda}}m_{\alpha}\left[G\right]^{\frac{1}{\alpha}}\,\phi_{\beta,q}\left[G\right]^{\frac{1}{\beta\lambda}}\label{eq:SecondCRuncertainty3}
\end{alignat}
\end{thm}

\begin{proof}
We first observe that for any complex valued function, we have 
\begin{equation}
\left(\frac{\partial|\psi(x)|}{\partial x_{i}}\right)^{2}=\frac{\partial\psi(x)}{\partial x_{i}}\frac{\partial\psi(x)^{*}}{\partial x_{i}}-|\psi(x)|^{2}\left(\frac{\partial\arg\psi(x)}{\partial x_{i}}\right)^{2}.\label{eq:dmodPsi}
\end{equation}
This relation implies that 
\begin{equation}
\left(\frac{\partial|\psi(x)|}{\partial x_{i}}\right)\leq\left|\frac{\partial\psi(x)}{\partial x_{i}}\right|,\label{eq:Inequality_for_dmodPsi}
\end{equation}
with equality if and only if $\arg\psi(x)=c,$ where $c$ is any constant.
Therefore, we see that we always have $\left\Vert \nabla_{x}|\psi(x,t)|\right\Vert _{*}\leq\left\Vert |\nabla_{x}\psi(x,t)|\right\Vert _{*}.$
In the following, we will take for $\left\Vert .\right\Vert _{*}$
a $\beta$-norm and for $\left\Vert .\right\Vert $ its dual $\alpha$-norm.
So doing, we get $\left\Vert |\nabla_{x}\psi(x,t)|\right\Vert _{\beta}^{\beta}=\sum_{i=1}^{n}|\partial_{i}\psi(x)|^{\beta}$.

At this point, we can invoke the sharp version of the Hausdorff-Young
inequality due to Babenko and Beckner \cite{beckner_inequalities_1975},
which states that for a pair of Fourier transforms $g$ and $\hat{g,}$
then for $1\leq\alpha\leq2$ 
\begin{equation}
\left\Vert g\right\Vert _{\beta}\leq\left(\frac{\beta^{\frac{1}{\beta}}}{\alpha^{\frac{1}{\alpha}}}\right)^{\frac{n}{2}}\,\left\Vert \hat{g}\right\Vert _{\alpha},\label{eq:BabenkoBecknerInequality}
\end{equation}
with equality if and only if $g$ is a Gaussian function or if $\alpha=\beta=2$
(Parseval's identity). Thus, it comes 
\begin{equation}
\int_{\mathbb{R}^n}\left\Vert \nabla_{x}|\psi(x,t)|\right\Vert _{\beta}^{\beta}\mathrm{d}x=\sum_{i=1}^{n}\int_{\mathbb{R}^n}|\partial_{i}\psi(x)|^{\beta}\mathrm{d}x\leq\left(\frac{\beta^{\frac{1}{\beta}}}{\alpha^{\frac{1}{\alpha}}}\right)^{\frac{n}{2}}(2\pi)^{\beta}\sum_{i=1}^{n}\left(\int_{\mathbb{R}^n}|\xi_{i}\hat{\psi}(\xi)|^{\alpha}\mathrm{d}\xi\right)^{\frac{\beta}{\alpha}}\label{eq:IneqOnGradientBeckner}
\end{equation}
 Inserting this inequality in (\ref{eq:FirstCRIntermsofPsi}), and
taking account of the fact that $\left\Vert x\right\Vert ^{\alpha}=\left\Vert x\right\Vert _{\alpha}^{\alpha}=\sum_{i=1}^{n}\left|x_{i}\right|^{\alpha}$,
we obtain
\begin{equation}
\frac{1}{M_{\frac{kq}{2}}[|\psi|^{2}]}\left(\sum_{i=1}^{n}\int_{\mathbb{R}^n}\left|x_{i}\right|^{\alpha}|\psi(x,t)|^{k}\mathrm{d}x\right)^{\frac{1}{\alpha}}\left(\sum_{i=1}^{n}\left(\int_{\mathbb{R}^n}|\xi_{i}|^{\alpha}|\hat{\psi}(\xi)|^{\alpha}\mathrm{d}\xi\right)^{\frac{\beta}{\alpha}}\right)^{\frac{1}{\beta}}\geq\frac{n}{2\pi kq}\left(\frac{\beta^{\frac{1}{\beta}}}{\alpha^{\frac{1}{\alpha}}}\right)^{-\frac{n}{2\beta}}\label{eq:FirstCRuncertainty1}
\end{equation}
Now, we can observe that by the H\"older inequality 
\begin{equation}
\left(\int_{\mathbb{R}^n}|\xi_{i}|^{\alpha}|\hat{\psi}(\xi)|^{\alpha}\mathrm{d}\xi\right)\leq\left(\int_{\mathbb{R}^n}|\xi_{i}|^{\beta}|\hat{\psi}(\xi)|^{\alpha}\mathrm{d}\xi\right)^{\frac{\alpha}{\beta}}\,\,\left(\int_{\mathbb{R}^n}|\hat{\psi}(\xi)|^{\alpha}\mathrm{d}\xi\right)^{1-\frac{\alpha}{\beta}}.\label{eq:BytheHolderinequality}
\end{equation}
Plugging in this inequality in (\ref{eq:FirstCRuncertainty1}), we
see that the exponent $\beta/\alpha$ simplifies, and that the inequality
can be written in a weaker but more symmetric form: 
\begin{equation}
\frac{M_{\frac{\alpha}{2}}[|\hat{\psi}|^{2}]^{\frac{1}{\alpha}-\frac{1}{\beta}}}{M_{\frac{kq}{2}}[|\psi|^{2}]}\left(\int_{\mathbb{R}^n}\left\Vert x\right\Vert _{\alpha}^{\alpha}|\psi(x,t)|^{k}\mathrm{d}x\right)^{\frac{1}{\alpha}}\left(\int_{\mathbb{R}^n}\left\Vert \xi\right\Vert _{\beta}^{\beta}|\hat{\psi}(\xi)|^{\alpha}\mathrm{d}\xi\,\,\right)^{\frac{1}{\beta}}\geq\frac{n}{2\pi kq}\left(\frac{\beta^{\frac{1}{\beta}}}{\alpha^{\frac{1}{\alpha}}}\right)^{-\frac{n}{2\beta}},\label{eq:FirstCRuncertainty2}
\end{equation}
which can also be written in terms of escort expectations as (\ref{eq:FirstCRuncertainty3}). 

For the second Cram\'er-Rao inequality, we follow the very same steps,
beginning with (\ref{eq:SecondCRIntermsofPsi}). By (\ref{eq:IneqOnGradientBeckner})
it comes 
\begin{equation}
\frac{1}{M_{\frac{k}{2}}[|\psi|^{2}]^{\frac{1}{k\lambda}}}\,\left(E_{\frac{k}{2}}\left[\left\Vert x\right\Vert _{\alpha}^{\alpha}\right]\right)^{\frac{1}{\alpha}}\left(\sum_{i=1}^{n}\left(\int_{\mathbb{R}^n}|\xi_{i}\hat{\psi}(\xi)|^{\alpha}\mathrm{d}\xi\right)^{\frac{\beta}{\alpha}}\right)^{\frac{1}{\beta\lambda}}\geq\frac{1}{\left(2\pi k\right)^{\frac{1}{\lambda}}}m_{\alpha}\left[G\right]^{\frac{1}{\alpha}}\,\phi_{\beta,q}\left[G\right]^{\frac{1}{\beta\lambda}}\left(\frac{\beta^{\frac{1}{\beta}}}{\alpha^{\frac{1}{\alpha}}}\right)^{-\frac{n}{2\beta\lambda}},
\end{equation}
which, by the H\"older inequality (\ref{eq:BytheHolderinequality}),
gives  (\ref{eq:SecondCRuncertainty3}). 
\end{proof}
In the general case, the inequalities (\ref{eq:FirstCRuncertainty3})
and (\ref{eq:SecondCRuncertainty3}) are not sharp, because they follow
from the Babenko-Beckner inequality (\ref{eq:BabenkoBecknerInequality})
and the H\"older inequality (\ref{eq:BytheHolderinequality}), where
the conditions for equality are not met simultaneously. However, we
can still get a sharp uncertainty relation, saturated by generalized
$q$-Gaussians, in the case $\alpha=\beta=2.$ Of course, for $q=1$ (which, with $\beta=2$ gives $k=2$), 
we obtain a multidimensional version of Heisenberg inequality, which
reduces to (\ref{eq:HeisenbergInequality}) in the scalar case. 
\begin{cor}
For $k=\beta/\left(\beta(q-1)+1\right),$ $\lambda=n(q-1)+1$, $q>\max\left\{ (n-1)/n,\, n/(n+\alpha)\right\} $
and $\gamma\geq2,$ $\theta\geq2$, the following uncertainty relations
hold:
\begin{alignat}{1}
 & \,\frac{M_{\frac{k}{2}}[|\psi|^{2}]^{\frac{1}{2}}}{M_{\frac{kq}{2}}[|\psi|^{2}]}\,\, E_{\frac{k}{2}}\left[\left\Vert x\right\Vert _{2}^{\gamma}\right]^{\frac{1}{\gamma}}E\left[\left\Vert \xi\right\Vert _{2}^{\theta}\right]^{\frac{1}{\theta}}\geq\frac{n}{2\pi kq},\label{eq:FirstInequalityForGeneralEuclideanMoments}\\
 & \frac{1}{M_{\frac{k}{2}}[|\psi|^{2}]^{\frac{1}{k\lambda}}}\,\left(E_{\frac{k}{2}}\left[\left\Vert x\right\Vert _{2}^{\gamma}\right]\right)^{\frac{1}{\gamma}}\left(E\left[\left\Vert \xi\right\Vert _{2}^{\theta}\right]\right)^{\mbox{\ensuremath{\frac{1}{\theta\lambda}}}}\geq\frac{1}{(2\pi k)^{\frac{1}{\lambda}}}m_{2}\left[G\right]^{\frac{1}{2}}\,\phi_{2,q}\left[G\right]^{\frac{1}{2\lambda}}.\label{eq:SecondInequalityForGeneralEuclideanMoments}
\end{alignat}
For $\gamma=\theta=2,$ the lower bound is attained if and only $|\psi(x)|^{k}/M_{\frac{k}{2}}[|\psi|^{2}]$
is a generalized $q$-Gaussian distribution, with $\arg\psi(x)=c$,
where $c$ a real constant. \textup{For $\frac{3}{2}-\frac{1}{\beta}>q,$
we also have 
\begin{equation}
\left(E_{\frac{k}{2}}\left[\left\Vert x\right\Vert _{2}^{\gamma}\right]\right)^{\frac{1}{\gamma}}\left(E\left[\left\Vert \xi\right\Vert _{2}^{\theta}\right]\right)^{\mbox{\ensuremath{\frac{1}{\theta\lambda}}}}\geq\frac{1}{(2\pi k)^{\frac{1}{\lambda}}}m_{2}\left[G\right]^{\frac{1}{2}}\,\phi_{2,q}\left[G\right]^{\frac{1}{2\lambda}}.\label{eq:WeakSecondInequality}
\end{equation}
}\end{cor}
\begin{proof}
Equality in (\ref{eq:Inequality_for_dmodPsi}) implies that $\arg\psi(x)=c$,
where $c$ a real constant. In the case $\alpha=\beta=2$, the Babenko-Beckner
inequality reduces to Parseval's equality, and the H\"older inequality
(\ref{eq:BytheHolderinequality}) is an identity. Therefore, we directly
obtain 
\begin{equation}
\frac{1}{M_{\frac{kq}{2}}[|\psi|^{2}]}\left(\int_{\mathbb{R}^n}\left\Vert x\right\Vert _{2}^{2}|\psi(x,t)|^{k}\mathrm{d}x\right)^{\frac{1}{2}}\left(\int_{\mathbb{R}^n}\left\Vert \xi\right\Vert _{2}^{2}|\hat{\psi}(\xi)|^{2}\mathrm{d}\xi\,\,\right)^{\frac{1}{2}}\geq\frac{n}{2\pi kq},
\end{equation}
which can also be written in terms of escort expectations as 
\begin{equation}
\,\frac{M_{\frac{k}{2}}[|\psi|^{2}]^{\frac{1}{2}}}{M_{\frac{kq}{2}}[|\psi|^{2}]}\,\, E_{\frac{k}{2}}\left[\left\Vert x\right\Vert _{2}^{2}\right]^{\frac{1}{2}}E\left[\left\Vert \xi\right\Vert _{2}^{2}\right]^{\frac{1}{2}}\geq\frac{n}{2\pi kq}.
\end{equation}
Similarly, the inequality (\ref{eq:SecondCRuncertainty3}) gives 
\begin{equation}
\frac{1}{M_{\frac{k}{2}}[|\psi|^{2}]^{\frac{1}{k\lambda}}}\,\left(E_{\frac{k}{2}}\left[\left\Vert x\right\Vert _{2}^{2}\right]\right)^{\frac{1}{2}}\left(E\left[\left\Vert \xi\right\Vert _{2}^{2}\right]\right)^{\mbox{\ensuremath{\frac{1}{2\lambda}}}}\geq\frac{1}{(2\pi k)^{\frac{1}{\lambda}}}m_{2}\left[G\right]^{\frac{1}{2}}\,\phi_{2,q}\left[G\right]^{\frac{1}{2\lambda}}
\end{equation}
In both inequalities, the lower bound is attained if $f(x)=|\psi(x)|^{k}/M_{\frac{k}{2}}[|\psi|^{2}]$
is a generalized $q$-Gaussian (which means that $|\psi(x)|^{2}$
is also a generalized $q$-Gaussian with a different entropic index:
$q'=1+k(q-1)/2$). 

By Jensen's inequality, we know that for $b\geq a$, we always have
$E[|X|^{b}]^{\frac{a}{b}}\geq E[|X|^{a}]$. Therefore, applying this
inequality with $\gamma\geq2,$ $\theta\geq2$, we obtain (\ref{eq:FirstInequalityForGeneralEuclideanMoments})
and (\ref{eq:SecondInequalityForGeneralEuclideanMoments}). Finally,
by the general power mean inequality, we know that $M_{a}[f]^{\frac{{1}}{a}}\geq M_{b}[f]^{\frac{1}{b}}$
for $a\geq b.$ Therefore, when $k>2,$ that is $\frac{3}{2}-\frac{1}{\beta}>q,$
then $M_{\frac{k}{2}}>M_{1}=1$, and the inequality (\ref{eq:SecondInequalityForGeneralEuclideanMoments})
yields (\ref{eq:WeakSecondInequality}). 
\end{proof}

\section{Conclusions}

This paper improves and build on our previous findings presented
in \cite{bercher_generalized_2012}. We connect concepts in estimation
theory to tools used in nonextensive thermostatistics and establish
general Cramér-Rao type inequalities valid for estimation purposes.
These results are given in the mutidimensional case, and a feature
of our approach is that it works for arbitrary norms on $\mathbb{R}^{n}.$
As a direct consequence, we obtain multidimensional versions of our
$q$-Cramér-Rao inequalities, which include the Barankin-Vajda as
well as the standard Cramér-Rao inequality as particular cases. Furthermore,
in the case of a translation family, we have shown that the corresponding
Cramér-Rao type inequality is saturated by multidimensional $q$-Gaussian
distributions. We have also presented a related general Cramér-Rao
inequality which is saturated by the same $q$-Gaussian distributions.
These results imply in particular that the generalized $q$-Gaussians
are the minimizers of an extended version of the Fisher information,
among all distributions with a given moment, just as the standard
Gaussian minimizes Fisher information over all distributions with
a given variance. Since these generalized Gaussian are already known
to be the maximum entropy distributions for Rényi or Tsallis entropies,
this yields a new, complementary, information theoretic characterization
of these distributions.  Finally, we have derived
new multidimensional uncertainty relations from the extended Cramér-Rao
inequalities. These uncertainty relations involve generalized expectations
computed with respect to escort distributions, and we have shown that
some of these uncertainty relations are saturated by generalized $q$-Gaussian
distributions, thus generalizing the fact that the standard Heisenberg
uncertainty relation is saturated by a standard Gaussian. 

 Some of the presented uncertainty inequalities are not sharp, and
it would certainly be of interest to try to improve them, by looking
to the extremal functions or to optimal constants. 
It would also be interesting to try to document the properties of the two generalized Fisher information.
Finally, in their recent
work \cite{lutwak_extensions_2012}, Lutwak et al. have introduced
an abstract, implicit, notion of Fisher information matrix attached
to a probability density. It would be of interest to examine whether
this notion could be extended and interpreted in the estimation theory
framework.
\vspace{-0.2cm}

\section*{Acknowledgments} 

\vspace{-0.2cm}
The author thanks the anonymous referees for their suggestions that helped improve the article. The author is also very grateful to Deane Yang (Polytechnic Institute, NY), who suggested to look at general norms, and to Matthieu Fradelizi (LAMA, Univ. Paris-Est), for helpful discussions related to this work. Thanks are extended to Lodie Garbell for her friendly proofreading of the manuscript. 

\section*{References}


\end{document}